\documentclass[12pt,onecolumn]{IEEEtran}

\usepackage[T1]{fontenc}
\usepackage{multirow}
\usepackage{cite}
\usepackage{graphicx,subfig}
\usepackage{psfrag}
\usepackage{amsmath,amssymb}
\usepackage{color}
\usepackage{tikz}
\usepackage{setspace}
\usepackage{pgfplots}
\usepackage{mathtools}

\def\xv{\boldsymbol{x}}
\def\hv{\boldsymbol{h}}
\def\E{\mathbb{E}}
\def\vv{\boldsymbol{v}}
\def\wv{\boldsymbol{w}}
\def\yv{\boldsymbol{y}}
\newtheorem{remark}{Remark}

\usepackage{amsmath,amssymb}
\newenvironment{proof}{\begin{IEEEproof}}{\end{IEEEproof}}

\newcommand{\defeq}{\triangleq}
\newtheorem{theorem}{Theorem}
\newtheorem{corollary}{Corollary}[theorem]

\newtheorem{example}{Example}

\ifCLASSINFOpdf
\else
\fi

\usepackage{amsmath}
\usepackage[cmintegrals]{newtxmath}

\begin{document}
\title{Adding transmitters can dramatically boost the finite-filesize gains of coded caching}
\title{Adding transmitters can dramatically boost the subpacketization-constrained DoF of coded caching}
\title{Adding transmitters can dramatically boost the subpacketization-limited receiver caching gains}
\title{Adding transmitters dramatically boosts subpacketization-limited receiver caching gains}
\title{Adding transmitters dramatically boosts receiver caching gains for finite file sizes}
\title{Adding transmitters dramatically boosts coded-caching gains for finite file sizes}

\author{Eleftherios~Lampiris and Petros~Elia
\thanks{The authors are with the Communication Systems Department at EURECOM, Sophia Antipolis, 06410, France (email: lampiris@eurecom.fr, elia@eurecom.fr).
The work is supported by the European Research Council under the EU Horizon 2020 research and innovation program / ERC grant agreement no. 725929.}
\thanks{Manuscript received December 10 2017.}
}

\maketitle

\doublespacing

\begin{abstract}
In the context of coded caching in the $K$-user BC, our work reveals the surprising fact that having multiple ($L$) transmitting antennas, dramatically ameliorates the long-standing subpacketization bottleneck of coded caching by reducing the required subpacketization to approximately its $L$th root, thus boosting the actual DoF by a \emph{multiplicative} factor of up to $L$. In asymptotic terms, this reveals that as long as $L$ scales with the theoretical caching gain, then the full cumulative (multiplexing + full caching) gains are achieved with constant subpacketization. This is the first time, in any known setting, that unbounded caching gains appear under finite file-size constraints. The achieved caching gains here are up to $L$ times higher than any caching gains previously experienced in any single- or multi-antenna fully-connected setting, thus offering a multiplicative mitigation to a subpacketization problem that was previously known to hard-bound caching gains to small constants.

The proposed scheme is practical and it works for all values of $K,L$ and all cache sizes. The scheme's gains show in practice: e.g. for $K=100$, when $L=1$ the theoretical caching gain of $G=10$, under the original coded caching algorithm, would have needed subpacketization $S_1 = \binom{K}{G}= \binom{100}{10} > 10^{13}$, while if extra transmitting antennas were added, the subpacketization was previously known to match or exceed $S_1$. Now for $L=5$, our scheme offers the theoretical (unconstrained) cumulative DoF $d_L = L+G = 5+10=15$, with subpacketization $S_L=\binom{K/L}{G/L} =\binom{100/5}{10/5} = 190$. The work extends to the multi-server and cache-aided IC settings, while the scheme's performance, given subpacketization $S_L=\binom{K/L}{G/L}$, is within a factor of 2 from the optimal linear sum-DoF.
\end{abstract}


%
\IEEEpeerreviewmaketitle

\section{Introduction}\label{sec:intro}
\IEEEPARstart{C}{oded} caching is a communication method invented in~\cite{MN14} that exploits receiver-side caches in broadcast-type communications, to achieve substantial throughput gains by delivering independent content to many users at a time. This method involves a cache placement phase and a delivery phase. During the placement phase, content from a library of files that are present at the transmitter, is properly pre-cached at the receiver caches. During the delivery phase --- which starts when users simultaneously request one desired library file each --- the transmitter encodes across different users' requested data content, in a way that creates multicasting opportunities even when users request different files.

Specifically the work in~\cite{MN14} considered the single-stream broadcast channel (BC) scenario where a single-antenna transmitter has access to a library of $N$ files, and serves $K$ receivers, each having a cache of size equal to the size of $M$ files. In a normalized setting where the link has capacity 1 file per unit of time, the work in~\cite{MN14} showed that any set of $K$ simultaneous requests can be served with normalized delay (worst-case completion time) which is at most $T = K(1-\gamma)/(1+K\gamma)$ where $\gamma \defeq M/N $ denotes the normalized cache size. This was a major breakthrough because it showed that an ever-increasing number of users can be served in finite time that converges to $T\approx \frac{1}{\gamma} = \frac{N}{M}$ as $K$ increases. This result implied a sum-DoF of
\[d_1(\gamma) = K(1-\gamma)/T = 1+K\gamma\] users served at a time. Given that in the absence of caching, only one user could be served at a time (because $d_1(\gamma=0) = 1$), the above implied a (theoretical) caching gain of
\[ G = d_1(\gamma)- d_1(\gamma=0) = K\gamma \]
representing the number of extra users that could be served at a time, additionally, as a consequence of introducing caching.

This massive theoretical gain came about because coded caching managed to remove the main inherent inefficiency of traditional caching methods, in which each receiver only ends up utilizing the cached fraction of just the one single file that that receiver had requested, while leaving all other information in the cache unused.
On the other hand, with coded caching, each receiver was now able to utilize the cached fraction of all $K$ requested files; The cached content of its own requested file provided the traditional local caching gain, while the cached content of the $K-1$ files requested by others, were now used to cancel the interference caused by those same files.

 This gain --- which is close to the theoretic optimal~\cite{MN14} --- was shown to persist under a variety of settings that include uneven popularity distributions~\cite{NiesenMtit17Popularity,ZhangLW15,JiTLC14}, uneven topologies~\cite{BidokhtiWT16isit,ZhangE16b}, a variety of channels such as erasure channels~\cite{GKY:15}, MIMO broadcast channels with fading~\cite{ZE:17tit}, a variety of networks such as heterogeneous networks~\cite{HachemKD15}, D2D networks~\cite{JiCM16D2D}\nocite{EE17}, and in other settings as well.

\subsection{Subpacketization bottleneck of coded caching\label{sec:IntroBottleneck}}
While though in theory, this caching gain $G = K\gamma$ increased indefinitely with increasing $K$, in practice the gain remained --- under most realistic assumptions --- hard-bounded by small constants, due to the fact that the underlying coded caching algorithms required the splitting of finite-length files into an exponential number of subpackets\footnote{Such high subpacketization originates from the fact that each file appears in each cache, and thus during delivery, a user must work together with all other users to get her file. This works --- at least in the original algorithm by Maddah Ali and Niesen --- by forming cliques of $K\gamma +1$ users, each requesting one subfile, where each user knows all subfiles requested from the clique, except the one that she herself requests. There are a total of $\binom{K}{K\gamma}$ cliques in which a specific user will have to be part of, and all of the cliques must be used; hence the need to split each file into $\binom{K}{K\gamma}$ different subfiles.}.
For the algorithm in \cite{MN14} in the original single-stream scenario, the near-optimal (and under some basic assumptions, optimal \cite{WanTP15,YuMA16}) gain of $G=K\gamma$, was achieved only if each file was segmented at least into a total of
\begin{equation} \label{eq:Fss}
S_1 = \binom{K}{K\gamma} \end{equation}
subpackets.
As a result, having a certain maximum-allowable subpacketization of $S_{max}$, implied that one could only encode over a maximum of
\begin{equation} \label{eq:barK}
\bar{K} = \arg\max_{K^o \leq K} \left\{ \binom{K^o}{K^o\gamma} \leq S_{max} \right\}
\end{equation}
users, which in turn implied a substantially reduced \emph{effective caching gain} $\bar{G}_1$ of the form
\begin{equation} \label{eq:EffectiveGainMN}
\bar{G}_1 = \bar{K}\gamma.
\end{equation}
Given that
\begin{equation} \label{eq:approximationSmn}
\binom{\bar{K}}{\bar{K}\gamma} \in \left[\left(\frac{1}{\gamma}\right)^{\bar{K}\gamma},\left(\frac{e}{\gamma}\right)^{\bar{K}\gamma}\right] = \left[\left(\frac{1}{\gamma}\right)^{\bar{G}_1},\left(\frac{e}{\gamma}\right)^{\bar{G}_1}\right]\end{equation}
this effective gain $\bar{G}_1$ was bounded as
\begin{equation}\label{barGc}
\frac{\log S_{max}}{1+\log\frac{1}{\gamma}} \leq \bar{G}_1  \leq \frac{\log S_{max}}{\log\frac{1}{\gamma}}, \ \ \ \ \bar{G}_1  \leq G
\end{equation}
($\log$ is the natural logarithm) which succinctly reveals that the effective caching gain $\bar{G}_1$ (and the corresponding \emph{effective sum-DoF} $\bar{d}_{1} \defeq 1+\bar{G}$) is placed under constant pressure from the generally small values\footnote{It is worth noting here that, as argued in~\cite{EJR:15}, in wireless cellular settings, the storage capacity at the end users is expected to induce $\gamma$ that can be less than $10^{-2}$, which --- for a given target caching gain --- implies the need to code over many users, which in turn increases subpacketization. Compounding on this problem, there is a variety of factors that restrict the maximum allowable subpacketization level $S_{max}$. One such parameter is the file size; for example, movies are expected to have size that is close to or less than 1 Gigabyte. Additionally, in applications like video streaming, a video file it self may be broken down into smaller independent parts (on which subpacketization will take place separately), in order to avoid the delay that comes from the asynchronous nature of decoding XORs in coded caching. Such restricted file sizes may be in the order of just a few tens of Megabytes. Another parameter that restricts $S_{max}$ is the minimum packet size; the atomic unit of storage is not a bit but a sector (newer `Advanced Format' hard drives use 4096-byte sectors and force zero-padding on the remaining unused sector), and similarly the atomic communication block is the packet, which must maintain a certain minimum size in order to avoid communication delay overheads.} of $\gamma$ and of $S_{max}$. This is reflected in Figure~\ref{fig:SccEffectiveGains0} and Figure~\ref{fig:SccEffectiveGains}.
\begin{figure}[ht!]
  \centering
\includegraphics[width=0.5\columnwidth, angle=270]{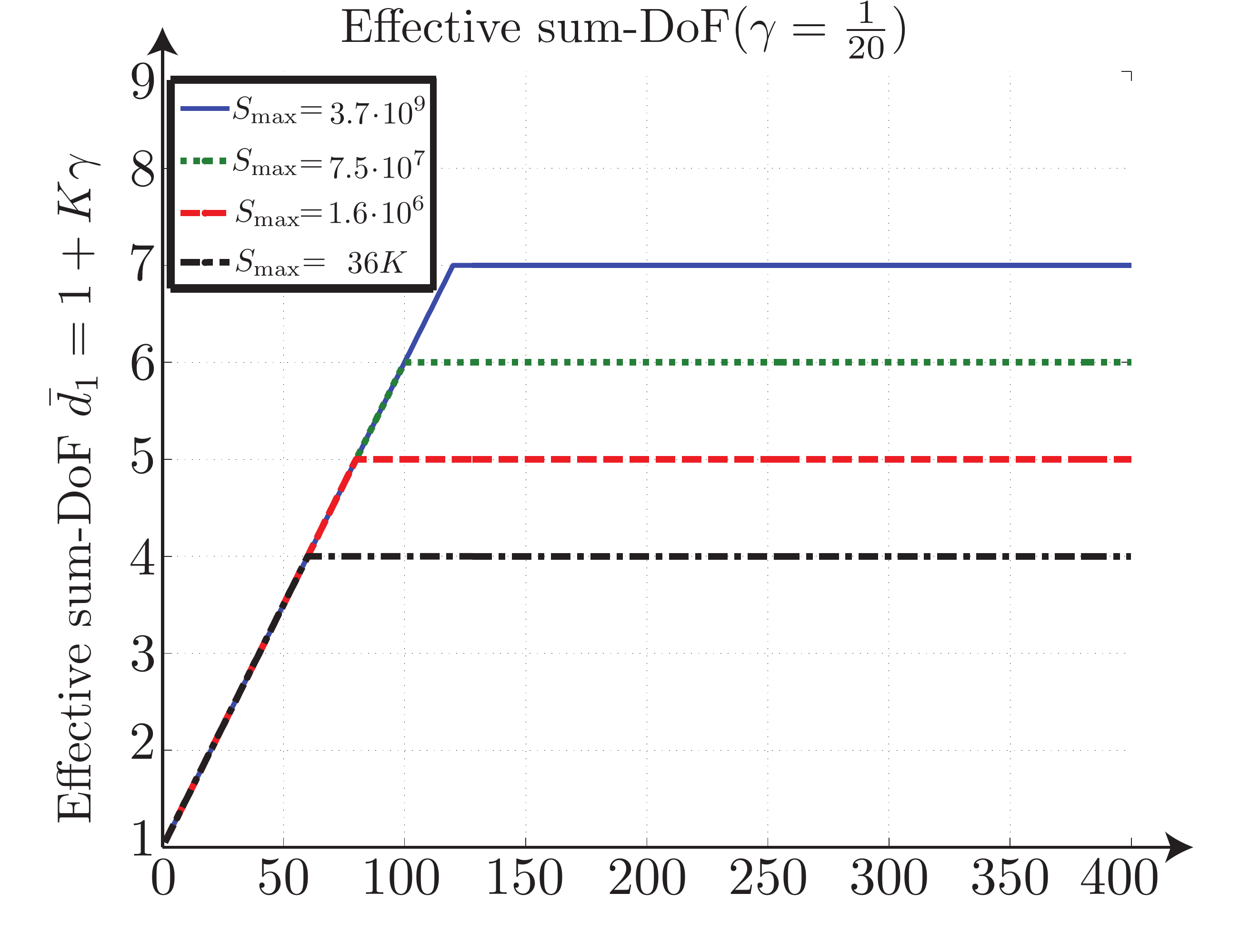}
\caption{Maximum effective DoF $\bar{d}_1$ achieved by the original centralized algorithm (single antenna, $\gamma = 1/20$) in the presence of different subpacketization constraints $S_{max}$. The gain is hard-bounded irrespective of $K$.}
\label{fig:SccEffectiveGains0}
\end{figure}
\begin{figure}[t!]
  \centering
\includegraphics[width=0.5\columnwidth, angle=270]{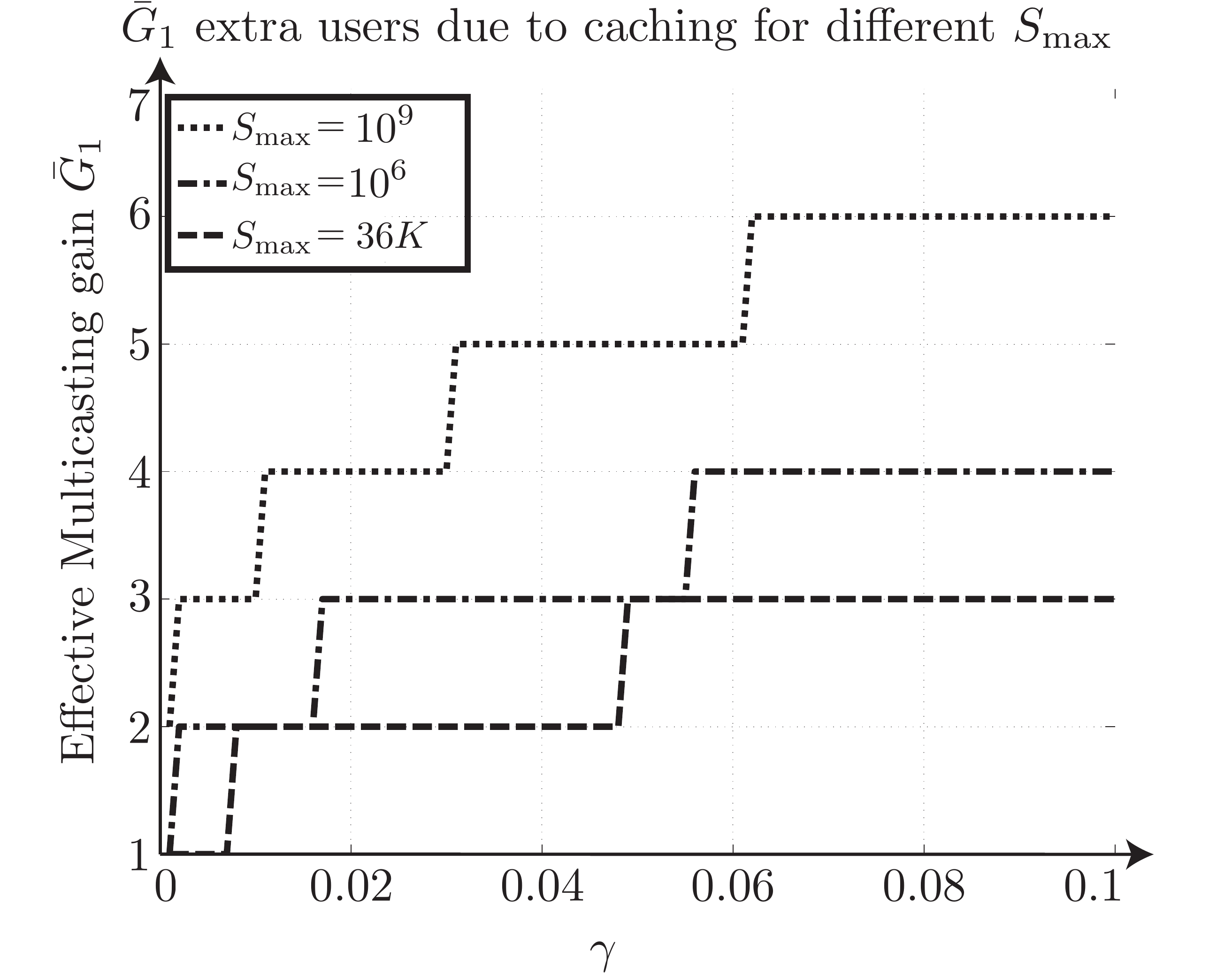}
\caption{Effective caching gain $\bar{G}_1 = \bar{d}_1-1$ (maximized over $K$) of the original algorithm for different $S_{max}$. Without subpacketization constraints, the theoretical gain is $G=K\gamma$ (unbounded as $K$ increases).}
\label{fig:SccEffectiveGains}
\end{figure}

\vspace{3pt}
\begin{example}\label{ex:ExLis1}
Looking at Figure~\ref{fig:SccEffectiveGains}, we see that if the library files (e.g. movies) are each of size 1 Gigabyte, and under a constraint that each packet cannot be less than 1 Kilobyte (KB) long (which jointly imply a subpacketization limit of $S_{max}\approx 10^6$), then having $\gamma<1/20$ would hard-bound the effective caching gain $\bar{G}_1$ to be less than $4$ (we add one extra in comparison to the plot, in order to account for any possible improvements from memory-sharing between operating points that yield neighboring integer-valued gains). This gain reduction is because we are forced to encode over less than $\bar{K} = 80$ users, to avoid a subpacketization $\binom{80}{4}>10^6$ that exceeds $S_{max}$. Having $\gamma<1/100$ would limit this gain $\bar{G}_1$ to be less than $3$ (since $\bar{K} = 300$ implies subpacketization $\binom{300}{3}>10^6$). When $S_{max} = 10^9$, where each packet consists of a single byte (without taking into consideration the overhead from using byte-sized packets), then having $\gamma<1/20$ would limit the effective gain to less than 6, while having $\gamma<1/100$ would limit the number $\bar{G}_1$ of additional users that could be served due to caching, to less than 4. When $S_{max} \approx 36K$, reflecting perhaps low-latency video streaming applications, for $\gamma\leq 1/20$ then $\bar{G}_1 \approx 3$ ($\bar{d}_{1} \approx 4 $ users at a time), while for $\gamma\leq 1/100$ then $\bar{G}_1 \approx 2$ ($\bar{d}_{1} \approx 3 $).
\end{example}
\vspace{3pt}

Similar conclusions were highlighted in \cite{ShanmugamJTLD16it}\nocite{JiSVLTC15}, in the context of decentralized coded caching algorithms (cf.~\cite{MND13}).

\paragraph*{New coded caching algorithms with reduced subpacketization}
This subpacketization bottleneck sparked significant interest in designing coded caching algorithms which can provide further caching gains under reduced subpacketization costs. A first breakthrough came with the work in \cite{yan2015placement} (see also \cite{tang2016coded}) which reformulated the coded caching problem into a \emph{placement-delivery array } (PD) combinatorial design problem, and which exploited interesting connections between coded caching and distributed storage to design an algorithm that provided a maximum theoretical caching gain of $G_{1,pd} = K\gamma-1$ (treating a total of $K\gamma$, rather than $K\gamma+1$, users at a time), at a reduced subpacketization of
\[S_{1,pd} = \left(\frac{1}{\gamma}\right)^{K\gamma-1} = \left(\frac{1}{\gamma}\right)^{G_{1,pd}}\]
thus allowing --- under some constraints on the operating parameters --- for an effective caching gain of
\begin{equation}\label{eq:RamQifaEffectiveGain}
\bar{G}_{1,pd}  = \min\left\{ \frac{\log S_{max}}{\log\frac{1}{\gamma}}, K\gamma-1\right\}.
\end{equation}

Similar conclusions were also drawn in \cite{tang2016coded} which used linear codes (LC) over high-order finite fields, to create set partitions that identify --- under some constraints on the values of $\gamma$ --- how the subpackets are cached and delivered, thus allowing for a tradeoff between an adjustable theoretical gain $G_{1,lc} \leq K\gamma-1$ and the corresponding subpacketization $S\approx \left(\frac{1}{\gamma}\right)^{G_{1,lc}}$, resulting in a similar effective gain of $\bar{G}_{1,lc}  \approx \frac{\log S_{max}}{\log\frac{1}{\gamma}}$ (naturally again the effective gain $\bar{G}_{1,lc}$ cannot exceed the theoretical gain $G_{1,lc}$).
Another breakthrough was presented in~\cite{ShangguanZG16} which took a hyper-graph theoretic approach to show that there do not exist caching algorithms that achieve a constant $T$ ($T$ is independent of $K$) with subpacketization that grows linearly\footnote{This assumes that $\gamma$ is independent of $K$, that each file is divided into an identical number of subpackets, and also assumes uncoded cache placement.} with $K$. This work also provided constructions which nicely tradeoff performance with subpacketization, which require though (Construction 6) that $K>4/\gamma^2$ (approximately) in order\footnote{$K$ must be large because the theoretical gain is reduced and is approximately $K\gamma^2/4$. $K$ must also be (essentially) a square integer; square integers become rarer as $K$ increases.} to have gains bigger than 1. Another milestone of a more theoretical nature was the very recent work in~\cite{shanmugam2017coded} which employed the Ruzsa-Szem\'{e}redi graphs to show for the first time that, under the assumption of (an unattainably) large $K$, one can get a (suboptimal) gain that scales with $K$, with a subpacketization that scales with $K^{1+\delta}$ for some arbitrarily small positive $\delta$.

While indeed different new algorithms provide exponential reduction in subpacketization, the corresponding improvement on the actual gain $\bar{G}$ --- over the original (MN) algorithm in~\cite{MN14}, for realistic values of $\gamma$ and $S_{max}$ --- remains hard bounded and small. For example, for $\gamma\leq 1/20$ and $S_{max} \leq 10^5$, no known algorithm can improve over the MN algorithm's effective caching gain (and effective DoF) by more than two\footnote{This best-known improvement is due to Construction 6 in~\cite{ShangguanZG16} ($a=b=2, \lambda = 40$) which encodes over $\bar{K}=3160$ users to give an effective sum-DoF of 6, while the MN algorithm gives a DoF of 4 (with $\bar{K}=60$).} (2 additional users served at a time) (see also Section~\ref{sec:resultsAlternateSchemes}).

\subsection{Coded caching with multiple transmitters \label{sec:IntroMultiAntenna}}
At the same time, different works (cf.\cite{ShariatpanahiMK16it,NMA:17TIT} as well as ~\cite{SenguptaTS15,CaoTXL16,HachemND16a,RoigTG17a,YangNK16a,ShariatpanahiCK17,PiovanoJC17,ZE:17tit,ZFE:15,ZhangE16asmallCaches}
and others) aimed at complementing such caching gains, with additional multiplexing gains that can appear when there are several transmitters.
One pioneering work in this direction was found in~\cite{ShariatpanahiMK16it} which considered a setting with $L = \lambda K, \ (\lambda\in(0,1))$ transmitters/servers
communicating (in the fully-connected BC context of a so-called `linear network' that can translate readily to a $K$-user wireless MISO BC with $L$ antennas) to $K$ single-antenna cache-aided receivers, and which provided a scheme that achieved a theoretical sum-DoF of \[d_{L}(\gamma) = L+K\gamma\] corresponding to a MIMO multiplexing gain of $L$ (users served, per second per hertz) and an additional theoretical caching gain of again $G=K\gamma$ (extra users served at a time, due to caching). This theoretical caching gain though was again restricted to an effective caching gain that was less than the effective gain $\bar{G}_{1}$ achieved in the single antenna case, because of a further increased subpacketization which now took the form
\begin{equation} \label{eq:Sms}
S = \binom{K}{K\gamma}\binom{K-K\gamma-1}{L-1}.
\end{equation}
While the subpacketization-constrained (effective) gains may have been reduced, this work in~\cite{ShariatpanahiMK16it} nicely showed that multiplexing and caching gains can in theory be combined additively.

Soon after, the work in \cite{NMA:17TIT} explored the scenario where coded caching involved both transmitter-side and receiver-side caches. In the context of a cache-aided interference scenario --- where $K_T$ transmitters with normalized cache size $\gamma_T$ (each transmitter could only store a fraction $\gamma_T$ of the entire $N$-file library), communicated to $K$ receivers with normalized cache size $\gamma$ --- the work provided a scheme that employed subpacketization \begin{equation} \label{eq:Sic}
S = \binom{K}{K\gamma}\binom{K_T}{K_T\gamma_T}\end{equation}
to achieve a sum-DoF of $\frac{K(1-\gamma)}{T} = K_T\gamma_T+K\gamma$ which was also proven to be at most a factor of 2 from the optimal (one-shot) linear-DoF. This nicely revealed that --- in the regime of unbounded subpacketization (unbounded file sizes) --- the cooperative multiplexing gain $K_T\gamma_T$ which is an outcome of the caching redundancy $K_T\gamma_T$ at the transmitter-side caches, can be additively combined with the theoretical caching gain $G = K\gamma$ attributed to receiver-side caching redundancy\footnote{By referring to transmitter-side redundancy $K_T\gamma_T$ and receiver-side redundancy $K\gamma$, we simply refer to the fact that each subfile resides in the caches of $K_T\gamma_T$ transmitters and in the caches of $K\gamma$ receivers.} $K\gamma$.
In both cases \cite{ShariatpanahiMK16it,NMA:17TIT}, the addition of the extra dimensions on the transmitter side, maintained the theoretical caching gains, added extra multiplexing gains, but maintained high subpacketization levels with generally reduced actual caching gains.

To the best of our knowledge, under the generous assumptions that $S_{max} \leq 10^5$, $\gamma \leq 1/50$ and $K\leq 10^5$, currently there exists no method \emph{in any known single-antenna or multi-antenna }fully connected setting, that allows for the introduction of more than $\bar{G} = 5$ additional users (per second per hertz, i.e., served at a time) due to caching\footnote{This corresponds to Construction 6 in~\cite{ShangguanZG16} ($a=b=2, \lambda = 100$), and it requires approximately 20000 users.}.

\subsection{Preview of results and paper outline}
Our contribution lies in the realization that having this extra dimensionality on the transmitter side, in fact reduces rather than increases subpacketization, and does so in a very accelerated manner. We will show a simple scheme for the multi-antenna/multi-node setting, that maintains the theoretical DoF \[d_L = L+G = L+K\gamma = K_T\gamma_T+K\gamma\] and does so with subpacketization
\begin{equation}\label{eq:ReducedSubpacketization1}
S_L = \binom{\frac{K}{L}}{\frac{K\gamma}{L}} = \binom{\frac{K}{K_T\gamma_T}}{\frac{K\gamma}{K_T\gamma_T}}
\end{equation}
which is approximately the $L$th root $S_L \simeq\sqrt[L]{S_{1}}$ of the original subpacketization $S_{1}=\binom{K}{K\gamma}$ corresponding to $L=1$. This will apply for all parameters $K,L,\gamma, K_T,\gamma_T$, it will imply very substantial subpacketization reductions even when $L$ is very small, as well as will imply that the theoretical DoF $d_L = L+K\gamma $ can be achieved with subpacketization $S_L = 1/\gamma = K/L$ when $L$ matches $K\gamma$.
The above expression~\eqref{eq:ReducedSubpacketization1} will imply a multi-antenna effective DoF
\[\bar{d}_L  = \min\{ L\cdot \bar{d}_1, d_L=L+K\gamma \}\] which is either $L$ times the single-antenna effective DoF $\bar{d}_1$, or it is the theoretical (unconstrained) $d_L=L+K\gamma$. In the end, we now know that having multiple antennas at the transmitter, not only provides a multiplexing gain, but also a multiplicative boost of the \emph{receiver-side} effective caching gain.

Finally, similar multiplicative boosts of the caching gain will be achieved when we apply the ideas here in conjunction with a variety of different underlying coded caching algorithms (see Section~\ref{sec:resultsAlternateSchemes}) like the ones in \cite{yan2015placement,tang2016coded}.

\paragraph*{Paper outline}

Section~\ref{sec:systemModel} elaborates on the system and channel model, Section~\ref{sec:Scheme} describes the scheme and presents simple examples of the construction, while Section~\ref{sec:mainResults} presents the main results which are accompanied with related examples of practical interest. The schemes and results are presented first for the integer case where $L | K$ and $L | K\gamma$ ($L$ divides $K$ and $K\gamma$), but we emphasize that the performance loss after removing the integer constraint, is very small (as we see in the appendix Section~\ref{sec:nonInteger}). Section~\ref{sec:conclusions} offers some conclusions, then the appendix Section~\ref{sec:IC} shows the details of how to adapt our approach to the cache-aided interference scenario with multiple independent cache-aided transmitters, while the appendix Section~\ref{sec:nonInteger} describes the slightly modified scheme for all $L,K$ when the assumptions $L | K$ and $L | K\gamma$ are removed.

\subsection{Notation}\label{sec:notation}
For clarity, we begin by recalling the common notation.
\begin{itemize}
\item $d_1(\gamma) = 1+K\gamma$ : Theoretical DoF ($L=1$)
\item $d_L(\gamma) = L+K\gamma$ : Theoretical DoF (multiple antennas)
\item $d_L(\gamma=0) = L$ : Multiplexing gain
\item $G$: Theoretical caching gain
    \begin{itemize}
    \item $G = d_1(\gamma) - d_1(\gamma=0) = d_L(\gamma) - d_L(\gamma=0) = K\gamma$
    \item $G$ additional users served at a time, due to caching\footnote{The choice here to measure the caching gain as the DoF difference $G = d_1(\gamma) - d_1(\gamma=0) = d_L(\gamma) - d_L(\gamma=0) = K\gamma$ rather than the DoF ratio, comes from the fact that in theory, the two gains (multiplexing and caching gains) appear to aggregate in an additive manner (this is discussed also in \cite{NMA:17TIT}). This choice of $G$ seems better suited for multi-antenna settings because a) it cleanly removes the multiplexing gain thus better isolating the true effect of caching, b) it reflects a caching gain that does not inevitably vanish with increasing $L$ (as would have happened had we used the DoF ratio), and c) it reflects a caching gain that scales with the cumulative cache size at the receiver side (i.e., scales with $K$).}
    \end{itemize}
\item $S_1 = \binom{K}{K\gamma}$: Subpacketization needed for theoretical $G$ ($L=1$)
\item $S_{max}$: Maximum allowable subpacketization
\end{itemize}
$\\[-15pt]$
\begin{itemize}
\item $S_L$: Subpacketization needed for theoretical $G$ (multiple antennas)
\item $\bar{d}_1(\gamma)$ : Effective (subpacketization-constrained) DoF ($L=1$)
\item $\bar{G}_1 = \bar{d}_1(\gamma) - 1$ : Effective caching gain ($L=1$)
\item $\bar{d}_L(\gamma)$ : Effective DoF (multiple antennas)
\item $\bar{G}_L = \bar{d}_L(\gamma) - L$ : Effective caching gain (multiple antennas)
\end{itemize}
In the above, $\bar{d}_L(\gamma = 0) = d_L(\gamma = 0) = L$ is the multiplexing gain, and $\bar{G}_L$ is the effective caching gain describing the actual number of additional users that can be served at a time as a result of introducing caching, under a subpacketization constraint. Finally the effective DoF $\bar{d}_L(\gamma) = L + \bar{G}_L$ describes the actual (total) number of users that can be served at a time, under a subpacketization constraint.

Furthermore we employ the following notation. $\mathbb{Z}$ will represent the integers, $\mathbb{Z}^{+}$ the positive integers, $\mathbb{R}$ the real numbers, and $\binom{n}{k}$ the $n$-choose-$k$ (binomial) operator. We will use $[K]\defeq \{1,2,\cdots,K\}$. If $\mathcal{A}$ is a set, then $|\mathcal{A}|$ will denote its cardinality. For sets $\mathcal{A}$ and $\mathcal{B}$, then $\mathcal{A} \backslash \mathcal{B}$ denotes the difference set. The expressions $\alpha | \beta$ (resp. $\alpha\nmid\beta$) denote that integer $\alpha$ divides (resp. does not divide) integer $\beta$.
Complex vectors will be denoted by lower-case bold font. We will use $||\xv||^2$ to denote the magnitude of a vector $\xv$ of complex numbers. Furthermore if $\mathcal{A}\subset [K]$ is a subset of users, then we will use $\mathbf{H}^{\mathcal{A}}$ to denote the overall channel from the $L$-antenna transmitter to the users in $\mathcal{A}$. Logarithms are of base~$e$. In a small abuse of notation, we will sometimes denote data sets the same way we denote the complex numbers (or vectors) that carry that same data.

\section{System and channel model\label{sec:systemModel}}
We initially consider the $K$-user multiple-input single-output (MISO) broadcast channel\footnote{We note that while the representation here is of a wireless model, the result applies directly to the multi-server setting of \cite{ShariatpanahiMK16it} with a fully connected linear network. We will also show at the end of this paper how the work here applies to the cache-aided interference scenario of \cite{NMA:17TIT}. Finally we note that in the DoF regime of interest, the single-antenna wireless setting ($L=1$) matches identically (in terms of the characteristics and performance) the original single-stream shared-link setting in \cite{MN14}.}, where an $L$-antenna transmitter communicates to $K$ single-antenna receiving users. The transmitter has access to a library of $N$ distinct files $W_1,W_2, \dots, W_N$, each of size $|W_n| = f$ bits. Each user $k \in \{1,2,\dots,K\}$ has a cache $Z_k$, of size $|Z_k| = Mf$ bits, where naturally $M \leq N$. Communication consists of the aforementioned \emph{content placement phase} and the \emph{delivery phase}. During the placement phase the caches $Z_1, Z_2, \dots, Z_K$ are pre-filled with content from the $N$ files $\{W_n\}_{n=1}^{N}$.

The delivery phase commences when each user $k$ requests from the transmitter, any \emph{one} file $W_{R_k}\in \{W_n\}_{n=1}^{N}$, out of the $N$ library files. Upon notification of the users' requests, the transmitter aims to deliver the (remaining of the) requested files, each to their intended receiver, and the challenge is to do so over a limited (delivery phase) duration $T$.
During this delivery phase, for each transmission, the received signals at each user $k$, will be modeled as
\begin{align}
y_{k}=\hv_{k}^{T} \xv + w_{k}, ~~ k = 1, \dots, K
\end{align}
where $\xv\in\mathbb{C}^{L\times 1}$ denotes the transmitted vector satisfying a power constraint $\E(||\xv||^2)\leq P$, where $\hv_{k}\in\mathbb{C}^{L\times 1}$ denotes the channel of user $k$ in the form of the random vector of fading coefficients that can change in time and space, and where $w_{k}$ represents unit-power AWGN noise at receiver $k$. We will assume that $P$ is high (high SNR), we will assume perfect channel state information throughout the (active) nodes as in \cite{ShariatpanahiMK16it,NMA:17TIT}, and we will assume that the fading process is statistically symmetric across users.

As in \cite{MN14}, $T$ is the number of time slots, per file served per user, needed to complete the delivery process, \emph{for any request}. The wireless link capabilities, and the time scale, are normalized such that one time slot corresponds to the optimal amount of time it would take to communicate a single file to a single receiver, had there been no caching and no interference
\footnote{As in the single-stream case in~\cite{MN14}, this achievable delay here is simply the minimum delay that allows, in the information theoretic sense (thus, under sufficiently long file sizes $f$), that each receiver $k$ decodes (with probability 1) its message $W_{R_k}$. $T$ reflects the maximum such minimum delay, maximized over all possible requests $\{W_{R_k}\}_{k=1}^K$. The high-SNR normalized delay $T$ (cf.~\cite{ZFE:15}; see also~\cite{SenguptaTS15,CaoTXL16}) used here, accounts for the file sizes and the high-SNR link capacity scaling $\log(\text{SNR})$, and is thus identical to the rate measure used in~\cite{MN14} for the single-stream error-free setting. Consequently in the high SNR setting of interest, an inversion leads to the equivalent measure of the cache-aided sum DoF $d_L(\gamma)=\frac{K(1-\gamma)}{T}$, as this is defined in~\cite{MNisit:15} in the context of transmitter-side caching, and in~\cite{ZFE:15} in the context of receiver-side caching (see also~\cite{SenguptaTS15,CaoTXL16}). The sum-DoF is the sum of multiplexing and theoretical caching gains, and -- as stated -- describes the total amount of users served at a time.}.

As in~\cite{MN14}, we will first consider the case where $\gamma = \frac{M}{N} = \{1,2,\cdots, K\}\frac{1}{K}$, while for non integer $K\gamma$, we will simply consider the result corresponding to $\lfloor K\gamma \rfloor$. Furthermore we will ignore the trivial case of $L\ge K(1-\gamma)$ which can be directly handled --- as shown in \cite{ShariatpanahiMK16it} --- to achieve the interference-free optimal $T = 1-\gamma$ corresponding to a sum-DoF $d_L(\gamma) = K$.

\section{Description of the scheme}\label{sec:Scheme}
We will present the scheme for all $K, \gamma,L$, first focusing here on the case where $L | K\gamma$ and $L|K$.

\paragraph{Grouping} We first split the $K$ users $k=1,2,\dots,K$ into $K'\defeq \frac{K}{L}$ disjoint groups
\[\mathcal{G}_g = \{  \ell K'+g, \ \ell = 0,1,\dots, L-1  \}, \ \text{for} \ g=1,2,\dots,K'  \]
of $|\mathcal{G}_g| = L$ users per group.
Our aim is to apply the algorithm of \cite{MN14} to serve $K'\gamma+1$ groups at a time, essentially treating each group as a single user.
Toward this, let \[\mathcal{T} =\{ \tau\in[K'] \ : \ |\tau| = K'\gamma\}\] be the set of
\begin{equation}\label{eq:cardinalityTau}
|\mathcal{T}| = \binom{K'}{K'\gamma}
\end{equation}
subsets in $[K']$, each of size $|\tau|=K'\gamma$, and let \[\mathcal{X} = \{ \chi\in[K'] \ : \ |\chi| = K'\gamma+1\}\] be the set of $|\mathcal{X}| = \binom{K'}{K'\gamma+1}$ subsets of size $|\chi|=K'\gamma+1$.
\paragraph{Subpacketization and caching}
We first split each file $W_n$ into $|\mathcal{T}|$ subfiles $\{W^\tau_n\}_{\tau\in \mathcal{T}}$, and then we assign each user $k\in \mathcal{G}_g$ the cache
\begin{equation} \label{eq:cachePlacement}
Z_k = Z_{\mathcal{G}_g} = \{ W_n^\tau \ : \ \forall \tau \ni g\}_{n=1}^N
\end{equation}
so that all users of the same group have an identical cache\footnote{A quick verification shows that \[|Z_{\mathcal{G}_g}| = N \frac{|\{\tau \in \mathcal{T}: g\in \tau\}|}{|\mathcal{T}|} = N \frac{\binom{K'-1}{K'\gamma-1}}{\binom{K'}{K'\gamma}} = N\gamma = M.\]}.

\paragraph{Transmission}
After notification of requests --- where each receiver $k$ requires file $W_{R_k}, \ R_k \in [N]$ --- the delivery consists of a sequential transmission $\{ \xv_\chi \}_{\chi\in \mathcal{X}}$ where each transmission takes the form
\begin{equation}\label{eq:transmission}
\xv_\chi = \sum_{g\in \chi} \sum_{k\in \mathcal{G}_g} W^{\chi \setminus g}_{R_k} \vv^{\mathcal{G}_g \setminus k}
\end{equation}
and where $\vv^{\mathcal{G}_g \setminus k}$ is an $L\times 1$ precoding vector that is designed to belong in the null space of the channel $\mathbf{H}^{\mathcal{G}_g \setminus k}$ between the $L$-antenna transmitter and the $L-1$ receivers in group $\mathcal{G}_g$ excluding receiver $k\in \mathcal{G}_g$.

\paragraph{Decoding --- `Caching-out' out-of-group messages} The corresponding received signal at user $k\in \mathcal{G}_g$ is then
\begin{equation}\label{eq:receivedSignal1}
\yv_{k,\chi} = \hv_{k}^{T} \xv_\chi + \wv_{k,\chi}
\end{equation}
and each such user $k\in \mathcal{G}_g$ can employ its cache to immediately remove all the files that are jointly undesired by its own group $\mathcal{G}_g$, i.e., receiver $k\in \mathcal{G}_g$ can remove \[\sum_{g'\in \chi\setminus g} \sum_{j\in \mathcal{G}_{g'}} W^{\chi \setminus g'}_{R_j} \vv^{\mathcal{G}_{g'} \setminus j}\] because $g'\neq g \in \chi$, i.e., because the cache of receiver $k$ includes all files $W^{\chi \setminus g'}_{R_j}$ in the above summation. This allows receiver $k$ to remove all files that are not of interest to its group $\mathcal{G}_g$, and thus to get
\begin{equation}\label{eq:receivedSignalAfterCachingOut}
\yv^{'}_{k,\chi} = \hv_{k}^{T} \bigl(   \sum_{j\in \mathcal{G}_g} W^{\chi \setminus g}_{R_j} \vv^{\mathcal{G}_g \setminus j} \bigr) + \wv_{k,\chi}.
\end{equation}
\paragraph{Nulling-out intra-group messages --- completion of decoding}
The interference for receiver $k$ now could only come from the files of the $L-1$ other users of its own group $\mathcal{G}_g$. This interference is averted directly by the ZF precoders (or any other DoF optimal precoder), and receiver $k$ can get the desired $W^{\chi \setminus g}_{R_k}$.

This is done instantaneously for all users $k\in \mathcal{G}_g$, and for all $g\in \chi$. Hence the scheme delivers to $K'\gamma+1$ groups at a time, thus to
\begin{equation}\label{eq:SchemeDoF}
d_{L}(\gamma) = L(K'\gamma+1) = K\gamma+L \end{equation} users at a time.
Then we do the same for another $\chi\in \mathcal{X}$. Along the different $\chi\in \mathcal{X}$, no subfile is repeated, and we can now conclude that the DoF is $K\gamma+L$, which as we saw (cf.~\eqref{eq:cardinalityTau}) is achieved here with subpacketization $S_L=\binom{K'}{K'\gamma}$.

\subsection{Example of scheme - alternate representation \label{sec:exampleOfScheme2}}
Let $K=50$, $L=5$ and $\gamma = M/N = 3/10$. We will achieve the sum-DoF of $d_\Sigma = L+G = L+K\gamma  = 5+15=20$, with a subpacketization of $120$.

First split the $K=50$ users into $K' = 10$ groups of $L=5$:
\begin{align*}
\mathcal{G}_1 = \{1,11,21,31,41\}, \dots , \mathcal{G}_{10}  = \{10,20,30,40,50\}.\end{align*}
Since $K'\gamma = 3$, we split each file $W_n$ into $|\mathcal{T}| = \binom{K'}{K'\gamma} = 120$ parts
\begin{align*}
 W_n = \{ W_n^{(1,2,3)},W_n^{(1,2,4)},\dots,W_n^{(1,3,4)},\dots, W_n^{(8,9,10)} \}
\end{align*}
and then fill the caches
\begin{align*}
Z_{\mathcal{G}_1} &= \{ W_n^{(1,2,3)},\!W_n^{(1,2,4)},\!\dots W_n^{(1,3,4)},\!\dots W_n^{(1,9,10)}\}_{n=1}^N\\
   & \vdots \\
Z_{\mathcal{G}_{10}} &= \{ W_n^{(1,2,10)},\!W_n^{(1,3,10)},\!\dots,W_n^{(2,3,10)},\!\dots W_n^{(8,9,10)}\}_{n=1}^N
\end{align*}
as described. We will serve $K'\gamma+1 = 4$ groups at a time. We treat the group clique $\chi = (1,2,3,4)$ first. Let
\[\wv_1^{(2,3,4)} = [W^{(2,3,4)}_{R_{1}},W^{(2,3,4)}_{R_{11}},W^{(2,3,4)}_{R_{21}},W^{(2,3,4)}_{R_{31}},W^{(2,3,4)}_{R_{41}}]^T\]
be the $L=5$ subfiles currently meant for the $5$ users in the first group. Similarly let $\wv_2^{(1,3,4)},\wv_3^{(1,2,4)},\wv_4^{(1,2,3)}$ be the $L$-length vectors of subfiles for the second, third and fourth groups respectively.
Then simply transmit
\begin{equation} \xv_{(1,2,3,4)} = (\mathbf{H}^{\mathcal{G}_1})^{-1} \wv_1^{(2,3,4)}  + (\mathbf{H}^{\mathcal{G}_2})^{-1} \wv_2^{(1,3,4)}  + (\mathbf{H}^{\mathcal{G}_3})^{-1} \wv_3^{(1,2,4)} +  (\mathbf{H}^{\mathcal{G}_4})^{-1}\wv_4^{(1,2,3)} \label{eq:trans1234}\end{equation}
where $(\mathbf{H}^{\mathcal{G}_g})^{-1}$ denotes the (normalized) inverse of the $L\times L$ channel to group $\mathcal{G}_g$.

Receiver 1 can immediately remove --- using its cache --- the last three summands in~\eqref{eq:trans1234}, and ZF can remove the unwanted $L-1=4$ elements from $\wv_1^{(2,3,4)}$. The achieved caching gain is $G = 15$, the sum-DoF is $d_L(\gamma) = 20$ (users at a time), and the subpacketization is $S_L=120$.

\section{Main results}\label{sec:mainResults}
We present the main results, first for the integer case where $L|K$ and $L|K\gamma$. The interpolation to all cases $K,L$ is easily handled using memory sharing, and as we note later on, does not result in substantial performance degradation. The details for this are handled in the appendix. We also try to highlight the practical relevance of some of these results, with examples.

We proceed with the main result.

\vspace{3pt}
\begin{theorem} \label{thm:SubpacketizationL}
In the cache-aided MISO BC with $L$ transmitting antennas and $K$ receiving users, the delay of $T = \frac{K(1-\gamma)}{L+K\gamma}$ and the corresponding sum-DoF $d_L(\gamma)=L+K\gamma$, can be achieved with subpacketization \[S_L = \binom{K/L}{K\gamma/L}.\]
\end{theorem}
\vspace{3pt}
\begin{proof}
The proof of this is direct from the description of the scheme. Specifically~\eqref{eq:cardinalityTau} tells us that the subpacketization is $\binom{K'}{K'\gamma}$ where $K' = K/L$, while~\eqref{eq:SchemeDoF} tells us that the DoF is $d_{L}(\gamma) = L(K'\gamma+1) = K\gamma+L$.
\end{proof}

\subsection{Effective gains and multiplicative boost of effective DoF \label{sec:ResultEffectiveGains}}
We recall that in the absence of subpacketization constraints, adding extra transmitting antennas, takes us from a theoretical sum-DoF $d_1 = 1+K\gamma$ to $d_L = L+K\gamma$ (cf.~\cite{ShariatpanahiMK16it}), leaving the theoretical caching gain unaffected, and adding $d_L(\gamma) - d_1(\gamma) = L-1$ DoF. For example, adding one extra antenna (going from $L=1$ to $L=2$) simply allows us to add one extra served user per second per hertz. What we will see here though is that, when subpacketization is taken into consideration, adding extra transmitting antennas (or later, adding extra transmitter-side caching) can have a much more powerful, multiplicative impact on the effective gains.

Recall from \eqref{eq:barK} that for $L=1$, the subpacketization takes the form $S_1 = \binom{K}{K\gamma}$, which -- as we briefly argued before -- means that having a maximum allowable subpacketization $S_{max}$, limits the number of users we can encode over, from $K$ to a smaller $\bar{K}_{1} \defeq  \arg\max\limits_{K^o \leq K} \left\{ \binom{K^o}{K^o\gamma} \leq S_{max} \right\}$. On the other hand, in the $L$ antenna case, the reduced subpacketization cost $S_L = \binom{\frac{K}{L}}{\frac{K\gamma}{L}}$ allows us, for the same constraint $S_{max}$, to encode over
\begin{equation} \label{eq:barKL}
\bar{K}_L  \defeq \arg\max_{K^o \leq K} \left\{ \binom{\frac{K^o}{L}}{\frac{K^o\gamma}{L}} \leq S_{max}\right\} = \min \{L\cdot \bar{K}_{1}, K\}
\end{equation}
users, just because the transition from $S_1$ to $S_L$ reflects a simple substitution of $K$ by $K/L$. Going from 1 to $L$ antennas, allows us to encode over $L$ times as many users (up to $K$), which in turn offers $L$ times more caching gain
\[\bar{G}_L = \min\{L\cdot \bar{G}_1, G\}\]
up to the theoretical $G = K\gamma$. Specifically if $\binom{\frac{K}{L}}{\frac{K\gamma}{L}} \leq S_{max}$ then $\bar{G}_L = G$ (corresponding to a multiplicative boost of $\frac{\bar{G}_L}{\bar{G}_1} = \frac{G}{\bar{G}_1}$), else the effective gain and the effective sum-DoF both experience a multiplicative increase by a factor of exactly $L$.
For completeness this is represented in the following corollary, which ignores for now integer rounding-off effects. The corollary follows directly from the above.

\vspace{3pt}
\begin{corollary}\label{cor:multiBoostMulticasting}
Under a maximum allowable subpacketization $S_{max}$, the multi-antenna effective caching gain and DoF take the form
\begin{align}\label{eq:MulticastingGains}
\bar{G}_L &= \min\{ L\cdot \bar{G}_1, G=K\gamma\}\\
\bar{d}_L & = \min\{ L\cdot \bar{d}_1, d_L=L+K\gamma\}
\end{align} which means that with extra antennas, the (single-antenna) effective DoF $\bar{d}_1$ is either increased by a multiplicative factor of $L$, or it reaches the theoretical (unconstrained) DoF $d_L=L+K\gamma$.
\end{corollary}

\vspace{3pt}

\begin{example}[Multiplicative boost of effective DoF]
In an $L$-antenna MISO BC, let $\gamma = 1/20$ and $K=1280$, corresponding to a theoretical caching gain of $G = K\gamma = 64$ and a theoretical sum-DoF of $d_L = L+G = L+64$. When $L=1$ then $d_1 = 65$, when $L=2$ then the sum DoF is $66$, and so on.
If the subpacketization limit $S_{max}$ was infinite, then of course the effective and theoretical caching gains would match, as we could get $\bar{G}_1 = \bar{G}_L = G= K\gamma = 64$ even when $L=1$, which would imply no multiplicative boost from having many antennas since $\frac{\bar{G}_L}{\bar{G}_1} = 1$.
If instead, the subpacketization limit was a lesser but still astronomical $S_{max} = \binom{K/2}{K\gamma/2} = \binom{640}{32}  $ then in the presence of a single antenna, we would encode over 640 users to get a constrained gain of $\bar{G}_1=640\frac{1}{20}=32$ which means that, irrespective of the number of antennas $L\geq 2$, the multiplicative boost would be $\frac{\bar{G}_L}{\bar{G}_1} = 2$.

Let us now consider a more reasonable $S_{max} = \binom{80}{4} \approx  1.5 \cdot 10^6$, and recall that for $L=1$, we could encode over only  $\bar{K} = 80$ users to get an effective caching gain $\bar{G}_1  = \bar{K} \gamma = 4$ treating a total of $\bar{d}_1 = 1+\bar{G}_1 = 5$ users at a time.
Assume now that we increased the number of transmitting antennas to $L=2$, in which case we would encode over $\bar{K}_L = L\cdot \bar{K}_1 = 2\cdot 80 = 160$ users which of course guarantees that $\binom{\frac{\bar{K}_L}{L}}{\frac{\bar{K}_L\gamma }{L}} = \binom{80}{4} \leq S_{max}$, and which yields a gain of $\bar{G}_L  = \bar{K}_L \gamma = 160 \frac{1}{20} = 8$, thus treating a total of $\bar{d}_L = L+\bar{G}_L  = L(1+\bar{G}_1) = 10$ users at a time, thus doubling the number of users served at a time, from 5 to 10.
Similarly for $L=4$, then $\bar{K}_L = 320$, which gives $\bar{G}_L  = 16$ and $\bar{d}_L = 20$, up until $L=16$ for which $\bar{K}_L = 16\cdot\bar{K}_1= 1280$, reaching the theoretical optimal $\bar{G}_L  = 64$, $\bar{d}_L = 80$, and the corresponding $16$-fold multiplicative DoF boost.

\end{example}

\vspace{3pt}
\begin{remark}
What we saw is that this $L$-fold multiplicative DoF boost stays into effect as long as $\binom{\frac{K}{L}}{\frac{K\gamma }{L}} \geq S_{max}$, so in essence it stays into effect as long as subpacketization remains an issue.
\end{remark}
\vspace{3pt}

The following corollary bounds the derived effective caching gain $\bar{G}_L$.

\vspace{3pt}
\begin{corollary} \label{cor:EffectiveGain}
Given a maximum allowable subpacketization $S_{max}$, the effective caching gain of the presented scheme is bounded as
\begin{equation}\label{eq:EffGainMN}
\bar{G}_L \geq \min\{ \ L\cdot \frac{\log S_{max}}{1+\log(\frac{1}{\gamma})},K\gamma \ \}.\end{equation} \end{corollary}
\vspace{3pt}

\begin{proof}
This follows directly from Sterling's approximation which bounds subpacketization as $S_L = \binom{K'}{K'\gamma} \leq \left( \frac{e}{\gamma}\right)^{K'\gamma} =  \left( \frac{e}{\gamma}\right)^{\frac{G}{L}}$ which directly implies that $\bar{G}_L \geq L\cdot \frac{\log S_{max}}{1+\log(\frac{1}{\gamma})}$ (up to the theoretical gain $G = K\gamma$).
\end{proof}
\vspace{3pt}

\begin{figure}[ht!]
  \centering
\includegraphics[width=0.4\columnwidth, angle=270]{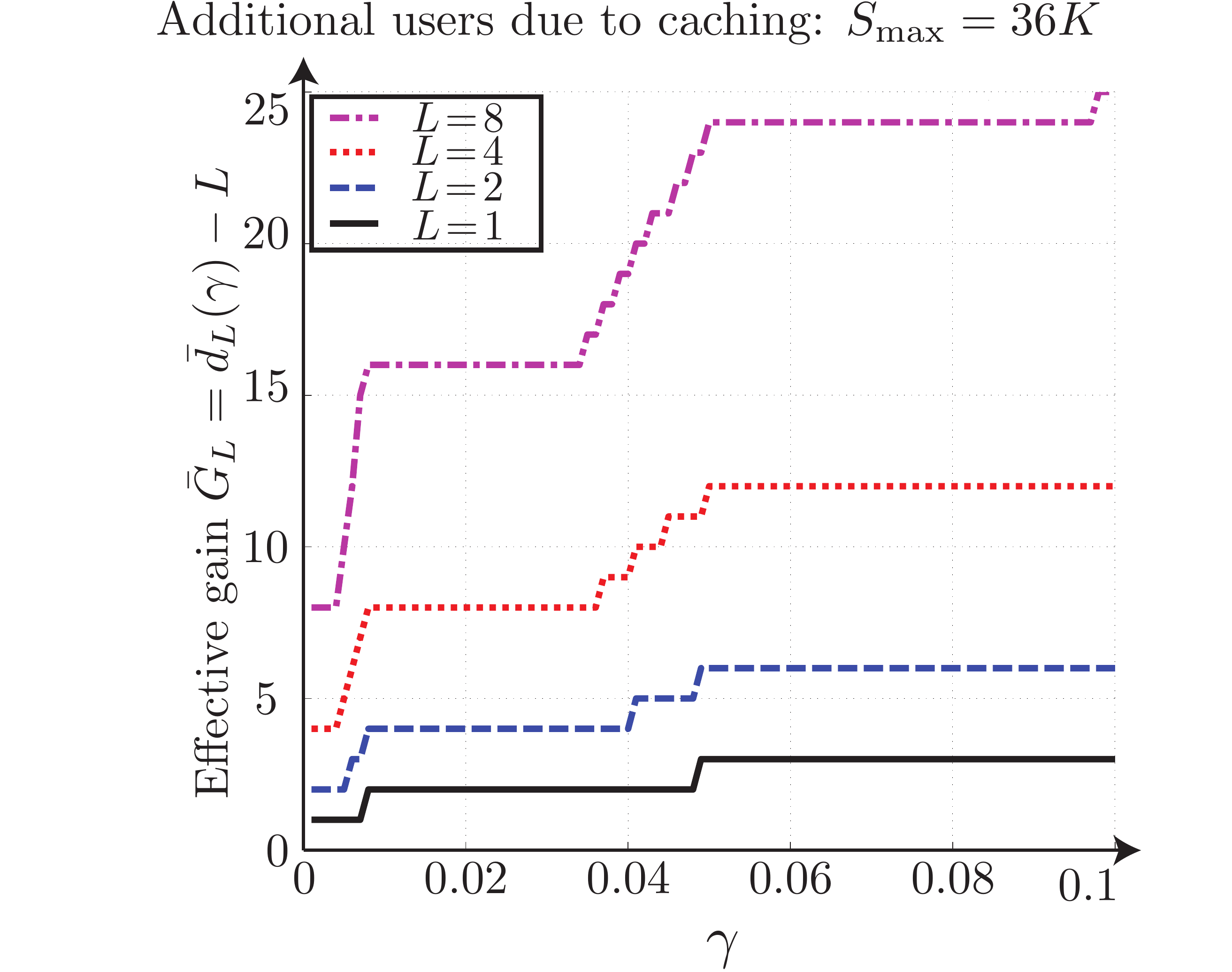}\\
\includegraphics[width=0.4\columnwidth, angle=270]{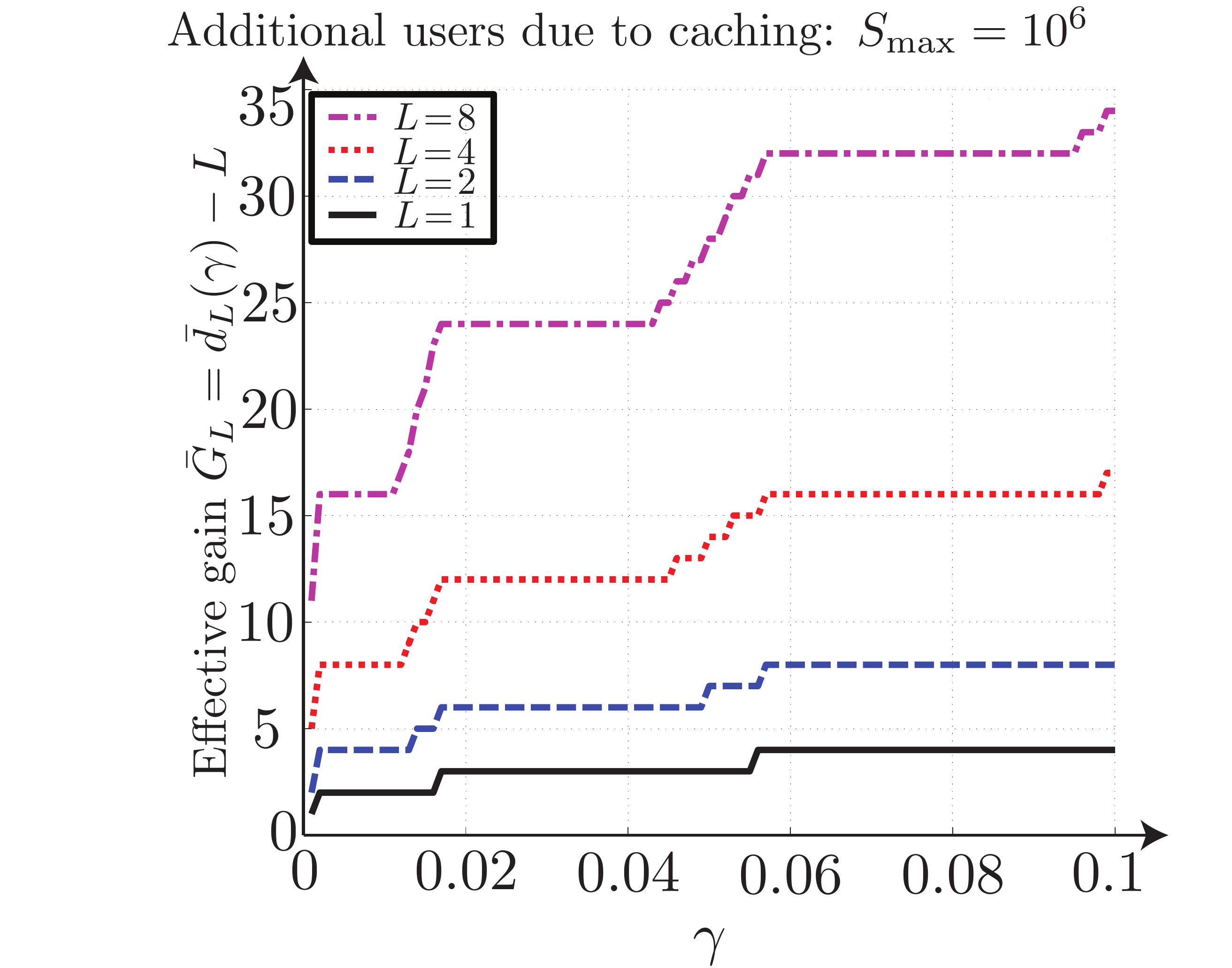}\\
\caption{
Maximum achievable effective caching gain $\bar{G}_L = d_L(\gamma)-L$ (maximized over all possible $K$), achieved by the new scheme for different $L$, under subpacketization constraint $S_{max} = 3.6\cdot 10^4$ (above) and $S_{max} = 10^6$ (below).}
\label{fig:EffectiveGainsComparison}
\end{figure}

\paragraph*{Practical implication - Making small caches relevant}
Another benefit of the reduced subpacketization here, is the resulting exponential increase in the range of cache sizes that can achieve a given target gain.
While in theory, a small $\gamma$ does not necessarily preclude higher caching gains because we could conceivably compensate by increasing the number of users we encode over, such an increase would increase subpacketization thus again precluding high gains (subpacketization limits would not allow for such an increase in the number of users we encode over). Specifically we recall (cf.~\eqref{eq:approximationSmn}) that when $L=1$ then the subpacketization is bounded as $S_1  \geq \left( \frac{1}{\gamma}\right)^{G}$, which means that to meet a subpacketization constraint $S_{\max}$ and a target caching gain of $G$, we need
\begin{equation} \label{eq:gammaThreshold}
\gamma \geq \left(S_{\max}\right)^{-1/G}.
\end{equation}
On the other hand, the reduced subpacketization $S_L  \geq \left( \frac{1}{\gamma}\right)^{\frac{1}{L}G}$ in the $L$ antenna case (cf.~\eqref{eq:approximationSmn}, after substituting $K$ by $K/L$), can allow for the same caching gain $G $ (given sufficiently many users to encode over) with only \begin{equation} \label{eq:gammaThreshold2}
\gamma \geq \left(\left(S_{\max}\right)^{-1/G}\right)^L.
\end{equation}
This exponential reduction in the minimum applicable $\gamma$, matches well the spirit of exploiting caches at the very periphery of the network, where we are expected to find relatively small but abundantly many caches.

\subsection{Subpacketization cost of complementing the multiplexing gains}

The following corollary highlights that, in an $L$-antenna MISO BC system, the subpacketization cost is not determined by $K$ or $L=\lambda K$, nor by the number of extra users $G$ we wish to add due to caching, but rather by the ratio $x = \frac{d_L(\gamma)}{d_L(\gamma=0)}$ between the DoF and the multiplexing gain.

\vspace{3pt}
\begin{corollary} \label{cor:multiplicativeInput}
In our $L$-antenna MISO BC setting, a subpacketization of \[S = \binom{1/\lambda}{x-1} = \binom{\frac{1}{\lambda}}{\frac{\gamma}{\lambda}}\] can yield a DoF that is $x$ times the multiplexing gain.
\end{corollary}
\vspace{3pt}
\begin{proof}
The DoF increase from $d_L(\gamma=0) = L$ to $d_L(\gamma) = L+K\gamma= x\cdot L, \ \ x\in\mathbb{Z}^{+}$, implies that $K\gamma = L(x-1)$ and that $\gamma = \lambda(x-1)$, which means that the corresponding subpacketization $S_L = \binom{K/L}{K\gamma/L}$ now takes the form $S = \binom{\frac{1}{\lambda}}{\frac{\gamma}{\lambda}} = \binom{1/\lambda}{x-1}$.
\end{proof}
\vspace{3pt}

\begin{remark}
This generalization here, from the known single-antenna case where $\lambda = 1/K, d_1(\gamma) = x = K\gamma+1$, to the $L$ antenna case, is nicely captured by Sterling's approximation which --- for $d(\gamma) = xL$ --- remains fixed at $S_L\in\left[ \left(\frac{1}{\gamma}\right)^{x-1}, \left(\frac{e}{\gamma}\right)^{x-1}\right]$. The result is simply a reflection of the fact that the same subpacketization cost of treating $K'\gamma+1$ users at a time ($K' = K/L$) in the single-antenna case, now guarantees the treatment of $K'\gamma+1$ groups at a time.
\end{remark}
\vspace{3pt}
\begin{example}
From the above we see that normalized cache sizes $\gamma = \lambda(x-1) = \lambda$ and a subpacketization $S_L = 1/\lambda = K/L$, suffice to double the total cache-free DoF ($x=2$), while $\gamma = 2\lambda$ and $S_L = \binom{1/\lambda}{2} < \frac{1}{2\lambda^2}$ can triple the number of users served at a time, from $L$ to $3L$. Hence for example in a cell of $K$ users served by a multi-antenna base-station that provides $d_L(\gamma=0)/K = L/K = \lambda = 1/30$ cache-free DoF per user, having $\gamma = \frac{xL-L}{K} = \lambda(x-1) = 2\lambda = 2/30$ and $S_{max} =\binom{1/\lambda}{x-1} = \binom{30}{2} = 435$ would allow caching to triple the number of users served at a time ($x=3$).
\end{example}
\vspace{3pt}

\subsection{Subpacketization scaling and algorithmic simplicity from matching multiplexing gain with caching gain}

Directly from the previous corollary, we also have the following.
\vspace{3pt}
\begin{corollary} \label{cor:multiplicativeInput2}
In asymptotic terms, as long as $L$ scales with the caching gain $K\gamma$, the entire sum-DoF $L+K\gamma$ is achievable with constant subpacketization.
\end{corollary}
\vspace{3pt}
\begin{proof}
As we have seen in the previous corollary, for $L = \frac{1}{q} K\gamma$ for some fixed $q\in \mathbb{Z}^{+}$, then the subpacketization is $S = \binom{1/\lambda}{q}$ and it is independent of $K,L$.
\end{proof}

An additional corollary is the following.
\vspace{3pt}
\begin{corollary} \label{cor:LisKgamma}
For $L=K\gamma$, the aforementioned DoF $L+K\gamma$ can be achieved with subpacketization $$S_L = \frac{1}{\gamma} = \frac{K}{L}.$$
\end{corollary}
\vspace{3pt}
The proof is direct from the above.

The following example highlights the utility of matching $K\gamma$ with $L$, and focuses on smaller cache sizes.
\vspace{3pt}
\begin{example}
In a BC with $\gamma=1/100$ and $L=1$, allowing for caching gains of $G = K\gamma = 10$ (additional users due to caching), would require $S_1 = \binom{1000}{10} > 10^{23}$ so in practice coded caching could not offer such gains. In the $L=10$ antenna case, this caching gain comes with subpacketization of only $S_L = K/L = 100$.
\end{example}
\vspace{3pt}

\subsection{Transmitter cooperation for boosting coded caching}
Until now we have explored the effect of having $L$ antennas at the transmitter. An identical effect will appear if instead of a single $L$-antenna transmitter, we consider $K_T$ independent single-antenna transmitters, each equipped with a cache of normalized cache size of $\gamma_T\geq \frac{1}{K_T}$ (as before, there are $K$ fully-interfering single-antenna receivers with normalized cache size $\gamma$). This setting corresponds to the $K_T\times K$ cache-aided interference scenario of~\cite{NaderializadehMA17a}, for which --- as discussed in Section~\ref{sec:IntroMultiAntenna} --- the (unconstrained) achieved `one-shot linear' sum-DoF takes the form $K_T\gamma_T+K\gamma$.

\vspace{3pt}
\begin{corollary} \label{cor:ICversion}
In the $K_T\times K$ cache-aided interference scenario with normalized cache sizes $\gamma_T,\gamma$, the sum-DoF of $K_T\gamma_T+K\gamma$, can be achieved with subpacketization of $$S_{K_T\gamma_T} = \binom{\frac{K}{K_T\gamma_T}}{\frac{K\gamma}{K_T\gamma_T}}.$$
\end{corollary}
\vspace{3pt}
\begin{proof}The constructive proof of the above is described in the Appendix.\end{proof}

\paragraph{Effects of cache-aided transmitter-cooperation on coded caching}
Given Corollary~\ref{cor:ICversion}, it is not difficult to conclude that all the previous corollaries apply directly to the $K_T\times K$ cache-aided interference scenario, after substituting $L$ with $K_T\gamma_T$. In particular, drawing from the previous corollaries, we can summarize the following results that apply to cache-aided transmitter cooperation.

\begin{itemize}
\item As the transmitter-side cache redundancy $K_T\gamma_T$ increases, the effective DoF will either be increased by a multiplicative factor of $K_T\gamma_T$, or it will reach the theoretical (unconstrained) DoF $K_T\gamma_T+K\gamma$ (cf. Corollary~\ref{cor:multiBoostMulticasting}).
\item In the presence of transmitter-side cache redundancy $K_T\gamma_T$, the effective caching gain is bounded below by $(K_T\gamma_T)\cdot \frac{\log S_{max}}{1+\log(\frac{1}{\gamma})}$ (cf.~Corollary~\ref{cor:EffectiveGain}).
\item Increasing the transmitter-side cache redundancy $K_T\gamma_T$, allows for an exponentially reduced minimum applicable $\gamma \geq \left(\left(S_{\max}\right)^{-1/G}\right)^{K_T\gamma_T}$ that can offer a (receiver-side) caching gain of $G = K\gamma$ (cf.~Section~\ref{sec:ResultEffectiveGains}).
\item Subpacketization $S = \binom{\frac{K}{K_T\gamma_T}}{x-1}$ can yield a sum DoF that is $x$ times the cooperative multiplexing gain $K_T\gamma_T$ (cf.~Corollary~\ref{cor:multiplicativeInput}).
\item In asymptotic terms, as long as the transmitter-side cache redundancy $K_T\gamma_T$ scales with the receiver cache redundancy $K\gamma$, the entire sum-DoF $K_T\gamma_T+K\gamma$ is achievable with constant subpacketization (cf.~Corollary~\ref{cor:multiplicativeInput2}).
\item When the transmitter-side and receiver-side cache redundancies match (i.e., when $K_T\gamma_T=K\gamma$), the DoF $K_T\gamma_T+K\gamma$ can be achieved with subpacketization $S_{K_T\gamma_T} = \frac{K}{K_T\gamma_T}$ (cf.~Corollary~\ref{cor:LisKgamma}).
\end{itemize}

\vspace{3pt}

\paragraph{Base-station cooperation for boosting coded caching}

The following corollary also holds.
\vspace{3pt}
\begin{corollary} \label{cor:ICversionMultiAntenna}
In the $K_T\times K$ cache-aided interference scenario with $\gamma_T\geq \frac{1}{K_T}$, if each transmitter has $L_T$ transmitting antennas, the sum-DoF of $K_TL_T\gamma_T+K\gamma$, can be achieved with subpacketization of $$S_{K_TL_T\gamma_T} = \binom{\frac{K}{K_TL_T\gamma_T}}{\frac{K\gamma}{K_TL_T\gamma_T}}.$$
Thus when $K_TL_T\gamma_T =K\gamma$ this sum-DoF can be achieved with subpacketization
\[S = \frac{K}{K_TL_T\gamma_T}.\]
\end{corollary}
\vspace{3pt}
The proof of the above is described briefly in the Appendix.

\vspace{3pt}
\begin{example}[Base-station cooperation]
Let us consider a scenario where in a dense urban setting, a single base-station ($K_T=1$) serves $K=10000$ cell-phone users, who are each willing to dedicate 20 Gigabytes of their phone's memory for caching parts from a Netflix library of $N=10000$ low-definition movies. Each movie is 1 Gigabyte in size, and the base-station can store 10 Terabytes. This corresponds to having $M = 20$, $\gamma = M/N = 1/500$, and $\gamma_T = 1$. If $L_T = 1$ (single transmitting antenna), a caching gain of $G = 20$ would have required (given the MN algorithm) subpacketization of $S_1 = \binom{K}{K\gamma} = \binom{10000}{20} > 10^{61}$.

If instead we had two base-stations ($K_T = 2$) with $L_T = 5$ transmitting antennas each, this gain would require subpacketization $S_L = \binom{\frac{K}{K_TL_T}}{\frac{K\gamma}{K_TL_T}} = \binom{10000/10}{20/10} = \binom{1000}{2} \approx 5\cdot 10^5$ (hence here, the introduction of caching would triple the total number of users served at a time), while with $K_T = 4$ such cooperating base-stations, this gain could be achieved with subpacketization of $\binom{10000/20}{20/20} = 500$.

If the library is now reduced to the most popular $N=1000$ movies (and without taking into consideration the cost of cache-misses due to not caching the tail of unpopular files), then the same 20 Gigabyte memory at the receivers would correspond to $\gamma = 1/50$ and to a theoretical caching gain of $G = K\gamma = 200$ additional users served per second per hertz. In this case, having a single large-MIMO array with $L_T = 100$ antennas, or having $K_T = 5$ cooperating base-stations with $L_T = 20$ antennas each, would yield a DoF $d_L(\gamma) = 300$ (caching would allow us to serve 200 additional users at a time), at subpacketization $S_L = \binom{10000/100}{200/100} = \binom{100}{2} \approx 5000$.
\end{example}
\vspace{3pt}

\subsection{Near-optimality of schemes}
The schemes that we have employed here (as described in Section~\ref{sec:Scheme} and in the Appendix) have the `one-shot, linear' property which means that each data element is manipulated linearly, and only once (a data bit is not transmitted more than once). This lends all the above results, except Corollary~\ref{cor:ICversionMultiAntenna}, amenable to the analysis in~\cite{NMA:17TIT} whose outer bound then allows us to directly conclude that the schemes are near optimal. This is described below, for purposes of completeness, in the form of a corollary.

\vspace{3pt}
\begin{corollary}\label{cor:nearOpt}
The described subpacketization $S_L = \binom{\frac{K}{L}}{\frac{K\gamma}{L}}$ and $S_{K_T\gamma_T}=\binom{\frac{K}{K_T\gamma_T}}{\frac{K\gamma}{K_T\gamma_T}}$ guarantees sum-DoF performance that is at most a factor of 2 from the theoretical optimal linear-DoF.
\end{corollary}
\begin{proof} As stated, the proof is direct from the bound in~\cite{NMA:17TIT}, from the performance achieved by the schemes here, and from the fact that the schemes have the `one-shot linear' property.
\end{proof}

\vspace{3pt}

We also note here that, to remove the integer constraints $L | K$ and $L | K\gamma$, we can readily use memory sharing as in \cite{MN14}. This is shown in the appendix, where we see that after removing the integer constraints, the results remain approximately the same except for a marginal increase in subpacketization
 to at most $S_L\leq
	L\cdot\max\left\{\binom{\lceil K/L\rceil}{\lceil K\gamma/L+1\rceil},\binom{\lceil K/L\rceil}{\lfloor K\gamma/L+1\rfloor}\right\}$, and a relatively small reduction in the achieved DoF ($d_L(\gamma) = L+K\gamma$) by a multiplicative factor (gap) that is bounded above by $\frac{5}{3}$ when $L>K\gamma$, and by $\frac{4}{3}$ when $L<K\gamma$ in which case the gap vanishes (converges to 1) as $K$ increases.

\subsection{Elevating different coded caching algorithms to the $L$ antenna setting \label{sec:resultsAlternateSchemes}}
The aforementioned subpacketization can be further reduced when considering alternate coded caching algorithms. We recall that the scheme that we have presented, involved `elevating' the original MN algorithm in~\cite{MN14}, from the single-stream scenario ($L=1$) with $K' = K/L$ users, to the $L$-antenna case with $K'$ groups of $L$-users per group. This same idea can apply \emph{directly} to other centralized coded caching algorithms like those in~\cite{yan2015placement,tang2016coded,shanmugam2017coded}, in which case the steps are almost identical:
\begin{itemize}
\item Choose the new coded caching algorithm for the single-stream $K'$-user scenario.
\item Split the $K$ users into $K'$ groups of $L$ users each, and employ the new algorithm to fill the caches as in the $K'$-user single-stream case, as if each group is a user, such that same-group users have caches that are identical.
\item Using the coded caching algorithm for the single-stream $K'$-user scenario, generate the sequence of XORs. Each XOR consists of $d^{'}_{1}(\gamma)$ summands, where $d^{'}_{1}(\gamma)$ is the theoretical sum-DoF provided by the coded caching algorithm in the $K'$-user single-antenna (single stream) BC.
\item Each element (summand) of the XOR, corresponds to a group of users, and each such XOR summand is replaced by a (precoded) $L$-length vector that carries the $L$-requests of the associated group. Add these $d^{'}_1$ vectors together, to form a composite transmitted vector that corresponds to the XOR.
\item Each composite vector treats a total of $d^{'}_1$ groups at a time, i.e., treats $L\cdot d^{'}_1(\gamma)$ users at a time.
\item Then continue with the rest of the XORs.
\end{itemize}
Hence we recall that when\footnote{We will henceforth use the term `elevate' to correspond to when we apply a single-stream coded caching algorithm to the multi-antenna case, via the above sequence of steps.} elevating the MN algorithm --- which, for the single-stream $K'$-user case, treats $d^{'}_1(\gamma) = K'\gamma+1$ users at a time --- we treated $d^{'}_1 = K'\gamma+1$ groups at a time, thus treating a total of $d_L(\gamma)=L\cdot d^{'}_1(\gamma) = L+K\gamma$ users at a time. On the other hand, when elevating for example the algorithms in~\cite{yan2015placement,tang2016coded}, we would naturally have to change the cache placement and the sequence of XORs, and we would have to account for the fact that --- for the single stream $K'$-user case --- the algorithm treats $d^{'}_{1,pd} = K'\gamma$ users at a time (not $K'\gamma+1$), and thus for $L\geq 1$, we would treat $d^{'}_{1,pd}  = \frac{K}{L}\gamma$ groups at a time ($L\leq K\gamma$), thus treating a total of $d_{L,pd}(\gamma)=L\cdot d^{'}_{1,pd}= K\gamma$ users at a time (not $K\gamma+L$).

The following corollary describes the effective caching gain provided by the scheme that elevates to the $L$ antenna case, the placement-delivery array (PD) and linear code (LC) algorithms in \cite{yan2015placement} and \cite{tang2016coded}. These algorithms exist under some constraints on $\gamma$.

\vspace{3pt}
\begin{corollary} \label{cor:EffectiveGainPD}
Given a maximum allowable subpacketization $S_{max}$, the effective caching gain of the here-elevated PD and LC algorithms, takes the form
\begin{equation}\label{eq:EffGainPDA}
\bar{G}_{L,pd} =\bar{G}_{L,lc} = \min\{ \ L\cdot \frac{\log S_{max}}{\log(\frac{1}{\gamma})} \ , K\gamma-L   \}.\end{equation}
\end{corollary}
\vspace{3pt}

\begin{proof}
With a theoretical gain $G_{L,pd} = d_{L,pd}(\gamma)-d_{L,pd}(\gamma=0)=K\gamma-L$, the underlying subpacketization $S_{L,pd} = \left( \frac{1}{\gamma}\right)^{K'\gamma - 1}$ can be written as $S_{L,pd} = \left( \frac{1}{\gamma}\right)^{\frac{G_{L,pd}}{L}}$, and thus the effective gain is $\bar{G}_{L,pd} = L\cdot \frac{\log S_{max}}{\log(\frac{1}{\gamma})} $, which is bounded by the theoretical caching gain $K\gamma-L$ offered by the scheme in the absence of subpacketization constraints.
\end{proof}
\vspace{3pt}

\paragraph{$L$-fold increase in impact of alternate coded caching algorithms\label{sec:increasedImpact}}
The fact that the underlying coded caching algorithm is used in our design at the level of groups of users, implies that any difference in the effective caching gain between two underlying algorithms in the single-stream case, will be magnified --- once each algorithm is elevated to the $L$-antenna case as was shown here --- by a factor of up to $L$. For example, if we were to compare the elevated MN scheme to, say, the aforementioned elevated PD and LC schemes, we would see (cf.~Corollary~\ref{cor:EffectiveGain} and Corollary~\ref{cor:EffectiveGainPD}) that
\begin{align*}
\bar{G}_{L,pd} & = \min\{ L\cdot \frac{\log S_{max}}{\log(\frac{1}{\gamma})}, K\gamma-L\} \\
\bar{G}_{L} & \geq \min\{ L\cdot \frac{\log S_{max}}{1+\log(\frac{1}{\gamma})}, \ K\gamma\}
\end{align*}
which would tell us that (when $K\gamma$ is an integer) the improvement in effective gains is bounded as
\[ \bar{G}_{L,pd} - \bar{G}_L \leq L\cdot \frac{\log S_{max}}{(\log(\frac{1}{\gamma}))(1+\log(\frac{1}{\gamma}))}.\]
When $L=1$, this improvement --- under realistic assumptions on $\gamma$ and $S_{max}$ --- can be small,
but when the algorithm is elevated to the multi-antenna setting, this improvement increases as a multiple of $L$.
\begin{remark}
This implies that the method proposed here, rather than bypassing the need for novel single-stream coded caching algorithms of reduced subpacketization, it in fact accentuates the importance of searching for such algorithms.
\end{remark}

\section{Conclusions}\label{sec:conclusions}
In the context of coded caching with multiple transmitting antennas (or with multiple transmitters or servers), we have presented a simple scheme which exploits transmitter-side dimensionality to provide very substantial reductions in the required subpacketization, without any sacrifice on the caching gain. As we have seen, this implies that, while in theory the addition of a few antennas provides an \emph{additive} sum-DoF increase of the form
\[ d_1(\gamma) = 1+ G \rightarrow d_L(\gamma) = L+G \]
(allowing the addition of $L-1$ extra users served at a time), in practice and in terms of subpacketization-constrained (effective) DoF, adding a few antennas implies a \emph{multiplicative} DoF increase of the form
\[ \bar{d}_1 = 1+\bar{G}_1 \rightarrow L+L\cdot\bar{G}_1. \]
Comparing the additive DoF increase of $L-1$ to the multiplicative DoF increase of $L$, suggests that for a large range of $L$, the main impact of multiple transmitting antennas is not the multiplexing gain, but rather the boost on the effect of receiver-side coded caching.

\subsection{Intuition on design}
The design was based on the simple observation that multi-node (transmitter-side) precoding, reduces the need for content overlap.
The subpacketization reduction from $\binom{K}{K\gamma}$ to $\binom{K/L}{K\gamma/L}$ was here related to the fact that the receivers of each group have identical caches. Subpacketization can generally increase because there needs to be a large set of pairings between the different caches.
Here the number of different distinct caches is reduced, and thus the number of such pairings remains smaller.

\subsection{Practicality and timeliness of result}
The scheme consists of the basic implementable ingredients of ZF and low-dimensional coded caching, and it works for all values of $K,L,\gamma,K_T,\gamma_T$. Its simplicity and effectiveness suggest that having extra transmitting antennas (servers) can play an important role in making coded caching even more applicable in practice, especially at a time when subpacketization complexity is the clear major bottleneck of coded caching, and also at a time when multiple antennas and transmitter cooperation are standard ingredients in wireless communications.

\paragraph{Separability between coded caching and PHY}
The result also advocates that some degree of joint consideration between cache-placement and network structure (here, for receiver-side cache-placement and `XOR' generation, we only need to know \emph{the number} of transmitters and receivers), can yield very substantial improvements in the effective DoF, as well as can maintain substantial (although certainly not complete) robustness to not knowing the exact network structure during the cache-placement phase.
While universal coded caching schemes that work obliviously of the structure of the communication network (cf.~\cite{NaderializadehMA17a}) carry an advantage when it comes to some robustness against network-structure uncertainty, the work here shows an instance where non-separated schemes have the potential to provide unboundedly better overall effective gains over universal schemes, by exploiting some of the structure of the network and by jointly considering coded caching and PHY.

\section{Appendix}

\subsection{Adapting to the cache-aided interference scenario}\label{sec:IC}

We now consider the cache-aided interference scenario studied in~\cite{NMA:17TIT}, with $K$ independent receivers, and with $K_T$ independent transmitters, where each transmitter has normalized cache size $\gamma_T = M_T/N$, where $fM_T$ is the size of each transmitter's cache. The scenario involves full connectivity (each receiver is connected to $K_T$ transmitters), and no information can be exchanged between the transmitters.

For transmitter-side cache placement, we ask that each subfile is placed at exactly $K_{T}\gamma_{T}$ transmitters, and to do so, we consecutively cache whole files into the transmitters, such that the first transmitter caches the first $M$ files, the second transmitter the next $M$ files, and so on, modulo $N$. Specifically, using $Z_{\text{Tx}_m}$ to denote the cache of transmitter $m\in [K_{T}]$, then the placement
\begin{align*}
	Z_{\text{Tx}_m}\!=\!\big{\{} W_{1+(n-1)\text{mod} N}: n\in \{1+(m-1)M,..., Mm \}\big{\}}
\end{align*}
guarantees the redundancy requirements and memory constraints. Now, for any given subfile, the $K_T\gamma_T$ transmitters that have access to this file, will employ CSIT in order to play the role of the aforementioned $L = K_T\gamma_T$ antennas, by precoding this said subfile using the exact same precoders described before, allowing for simultaneous separation of the $L=K_T\gamma_T$ streams within any given group $\mathcal{G}_g$ of $L=K_T\gamma_T$ receivers. As before, the aforementioned caching allows for treatment of $K'\gamma+1$ groups at a time, and a treatment of $K_T\gamma_T+K\gamma\leq K$ users at a time (Corollary~\ref{cor:ICversion}).

Finally it is easy to see that the above idea holds directly for the case where --- in the above $K_T\times K$ cache-aided interference scenario with $\gamma_T\geq \frac{1}{K_T}$ --- each transmitter has $L_T$ transmitting antennas. In this case we can see that this same placement method has the desired property that each subfile is available at $L=K_TL_T\gamma_T$ antennas, yielding a sum-DoF of $K_TL_T\gamma_T+K\gamma$ which can be achieved with subpacketization $$S_{K_TL_T\gamma_T} = \binom{\frac{K}{K_TL_T\gamma_T}}{\frac{K\gamma}{K_TL_T\gamma_T}}$$ as mentioned in Corollary~\ref{cor:ICversionMultiAntenna}.

\subsection{General scheme: removing the integer constraint}\label{sec:nonInteger}

We proceed to remove the constraints $L | K$ and $L | K\gamma$, by applying as in \cite{MN14} memory sharing. The results, after removing the integer constraints, will remain approximately the same except for a marginal increase in subpacketization\footnote{Note that for the settings in \cite{MN14,ShariatpanahiMK16it,NMA:17TIT}, the aforementioned subpacketization costs in~\eqref{eq:Fss},\eqref{eq:Sms} and~\eqref{eq:Sic} do not account for the extra subpacketization costs due to memory sharing.}
 to at most $S_L\leq
    K\cdot\max\left\{\binom{\lceil K/L\rceil}{\lceil K\gamma/L+1\rceil},\binom{\lceil K/L\rceil}{\lfloor K\gamma/L+1\rfloor}\right\}$ and a relatively small reduction in the achieved DoF ($d_L(\gamma) = L+K\gamma$) by a multiplicative factor (gap) that is bounded above by $2$ when $L>K\gamma$ and by $\frac{3}{2}$ when $L<K\gamma$, while the gap vanishes as $\frac{K\gamma}{L}$ increases.

To remove the constraint $L | K$ we will add to the system phantom users such that the new (hypothetical) number of users is $\hat{K}=L\left\lceil \frac{K}{L}\right\rceil$. Moreover, if $L\nmid \hat{K}\gamma$ we will perform memory sharing (cf.~\cite{MN14}) by splitting each file $W_n$ into two parts, $W_{n}', W_{n}''$ of different sizes $|W_{n}'| = p|W_{n}|$ and $|W_{n}''| = (1-p)|W_{n}|$, and cache each part with normalized cache sizes $\gamma'=\frac{|Z_k\cap W_{n}'|}{|W_{n}'|}= \frac{L}{\hat{K}}\left\lfloor \frac{\hat{K}\gamma}{L}\right\rfloor$ and $\gamma''=\frac{|Z_k\cap W_{n}''|}{|W_{n}''|}=\frac{L}{\hat{K}}\left\lceil \frac{\hat{K}\gamma}{L}\right\rceil$, which guarantees that $L | \hat{K}\gamma'$ and $L | \hat{K}\gamma''$. This also gives that $p=\frac{\gamma''-\gamma}{\gamma''-\gamma'}$.

Then, as the original scheme describes, we divide $W_{n}'$ into $\binom{\hat{K}/L}{\hat{K}\gamma'/L}$ parts, $W_{n}''$ into $\binom{\hat{K}/L}{\hat{K}\gamma''/L}$ parts, and cache from $W_{n}',W_{n}''$ according to~\eqref{eq:cachePlacement}. The corresponding subpacketization cost is thus bounded as
\begin{align}
    S  & \le K\cdot\max\left\{\binom{\hat{K}/L}{\hat{K}\gamma'/L},\binom{\hat{K}/L}{\hat{K}\gamma''/L}\right\} \nonumber \\ & \le K\cdot\max\left\{\binom{\lceil K/L\rceil}{\lceil K\gamma/L\rceil+1},\binom{\lceil K/L\rceil}{\lfloor K\gamma/L\rfloor+1}\right\}
\end{align}
where the multiplicative factor of $K$ is the one that upper bounds the subpacketization effect of splitting the file in two parts before subpacketizing each part. This effect is bounded by $K$ because $p\geq 1/K$ by virtue of the fact that $K\gamma$ is an integer\footnote{To see this, we rewrite $\gamma$ as $\gamma=a/K$ where $a$ is an integer, and then we see that
$p = \frac{\gamma''-\gamma}{\gamma''-\gamma'}=\frac{\left\lceil\frac{\hat{K}a}{KL}\right\rceil-\frac{a\hat{K}}{KL}}{\left\lceil\frac{\hat{K}a}{KL}\right\rceil-\left\lfloor\frac{\hat{K}a}{KL}\right\rfloor} >\frac{1}{K}$
where, in the last step we used the fact that the denominator is 1 (unless it is zero, in which case there is no additional subpacketization cost), while for the numerator we have that $\left \lceil\frac{\hat{K}a}{KL}\right\rceil-\frac{a\hat{K}}{KL}>\frac{1}{K}$ because $L| \hat{K}a$.}.

Then, in order to derive a multiplicative gap on DoF, $d_L^{nc}$, that accounts for removing the two constraints, we will consider two separate cases. 
First, we will look at the case of $\hat{K}\gamma\le L$. By applying memory sharing, we can see that each part will be cached with redundancy $0$ and $L$ respectively. This means that the completion time will be $T=\frac{m'}{0+L}+\frac{m''}{L+L}$, where $m'=Kp(1-\gamma')$ and $m''=K(1-p)(1-\gamma'')$. Then, we can see that the completion time is upper-bounded $T\le \frac{K(1-\gamma)}{L}$ and lower-bounded $T\ge \frac{K(1-\gamma)}{2L}$, which incorporates the facts that the performance cannot be worse than if there was no caching gains, but it cannot be better than if the caching gain was $L$. Using that, we can calculate the bounds of the DoF as follows
\begin{align*}
	\frac{K(1-\gamma)}{L}&\ge T\ge \frac{K(1-\gamma)}{2L}\\
	\frac{K(1-\gamma)}{\frac{K(1-\gamma)}{2L}}&\ge d_L^{nc}\ge \frac{K(1-\gamma)}{\frac{K(1-\gamma)}{L}}\\
	2L&\ge d_L^{nc}\ge L
\end{align*}
which implies a gap of 2.

Similarly, for $K\gamma\in (qL, qL+1),~~ q=\{1,2,...\}$ we can see that the above gap becomes $\frac{q+1}{q}$.


\bibliographystyle{IEEEtran}

	\end{document}